\newcommand{\ket}[1]{\mbox{$ | #1 \rangle $}}
\newcommand{\bra}[1]{\mbox{$ \langle #1 | $}}
\newcommand{\prob}[1]{\mbox{$ P \left( #1 \right) $}}
\newcommand{\Id}{\mathds{1}}

\documentclass{iopart}

\usepackage{cite}
\expandafter\let\csname equation*\endcsname\relax
\expandafter\let\csname endequation*\endcsname\relax
\usepackage{amsmath}
\usepackage{amssymb}
\usepackage{amsthm} %theorem environment
\usepackage{amsfonts}
\usepackage{dsfont} % for the identity
\usepackage{setspace}
\usepackage{graphicx}

\newtheorem{Lemma}{Lemma}

\begin{document}

\title{A new device-independent dimension witness and its experimental implementation}
\author{Yu Cai}
\address{Centre for Quantum Technologies, National University of Singapore, Singapore}

\author{Jean-Daniel Bancal}
\address{Centre for Quantum Technologies, National University of Singapore, Singapore}

\author{Jacquiline Romero\footnote{Presently at: School of Mathematics and Physics, University of Queensland, St Lucia 4072, Australia}}
\address{School of Physics and Astronomy, SUPA, University of Glasgow, United Kingdom}--

\author{Valerio Scarani}
\address{Centre for Quantum Technologies, National University of Singapore, Singapore}
\address{Department of Physics, National University of Singapore, Singapore}

\begin{abstract}
A dimension witness is a criterion that sets a lower bound on the dimension needed to reproduce the observed data. Three types of dimension witnesses can be found in the literature: device-dependent ones, in which the bound is obtained assuming some knowledge on the state and the measurements; device-independent prepare-and-measure ones, that can be applied to any system including classical ones; and device-independent Bell-based ones, that certify the minimal dimension of some entangled systems. %Only the first two types have been implemented so far. 
Here we consider the Collins-Gisin-Linden-Massar-Popescu (CGLMP) Bell-type inequality for four outcomes. We show that a sufficiently high violation of this inequality witnesses $d\geq 4$ and present a proof-of-principle experimental observation of such a violation. This presents a first experimental violation of the third type of dimension witness beyond qutrits.
\end{abstract}

\maketitle

\section{Introduction}

The \textit{dimension} of a physical system is the number of its perfectly distinguishable states. As such, it is the most basic quantifier of the capacity of that system to encode information. In classical physics, the dimension coincides with the number of possible pure states. In quantum physics, coherent transformations allow for the creation of infinitely many pure states even in the case of finite dimension, a possibility that lies at the heart of quantum information processing.

Arguably, all physical systems have infinite dimension: the electron that carries a spin has also a wave function, the electromagnetic field has potentially infinitely many modes, each being infinite-dimensional. Nevertheless, it is meaningful to ask the question: \textit{what is the minimal dimension one must be able to address, in order to produce some data?} The question can be asked from two different perspectives. One can see it as an upper bound on necessity: ``with suitable control on systems of dimension $d$, you can produce the data you want''. This viewpoint is suited for designing an implementation of a protocol, and is of course the way complexity theorists look at it. Alternatively, one can see it as a lower bound on sufficiency, which is more suited as a comment to the performance of a setup (``having observed these data, I know that my setup is addressing at least $d$ dimensions''). To avoid confusions, this paper is consistently written from the latter perspective.

A \textit{dimension witness (DW)} is a test that provides such a lower bound on the required dimension. The simplest DW is suggested by the definition of dimension itself: one tries to encode, then decode faithfully one \textit{Dit} of information. If the decoding is free from error, the carrier must have dimension $d\geq D$. Since nothing has to be specified about the coding and decoding, this basic test is already device-independent. However, it does not single out quantum systems: the same lower bound is obtained whether the information is encoded in a classical or quantum carrier. More elaborated encode-and-decode (prepare-and-measure) DWs can certify that $d\geq D_c$ in presence of classical manipulations, $d\geq D_q$ in presence of quantum ones, with $D_c > D_q$~\cite{DW1gallego,DW1brunner,DW1bowles}. Such DWs are handy because they can be used to bound the dimension of the systems produced by a single source; they have already been implemented in various ways~\cite{DW1exp1,DW1exp2,DW1exp2b,DAmbrosio2014,Tavakoli}. 

Device-independent DW can be based on the violation of some Bell inequalities~\cite{DW2original, DW2vertesipal, Moroder,Navascues2014Characterization,Navascues2015Bounding} or other non-linear criteria that detect nonlocality~\cite{wehner,SikoraPSD}. Such a violation is impossible with classical no-signalling resources and actually requires entanglement. Thus, these DWs certify how-large-dimensional entanglement is needed in order to reproduce the observed data. The first such entanglement witness used the Collins-Gisin-Linden-Massar-Popescu Bell-type inequality~\cite{CGLMP} for three outcomes (CGLMP$_3$): it showed that, if the violation is high enough, $d\geq 3$ can be certified~\cite{DW2original}; in fact the data of~\cite{vaziri2002} has shown such a violation. Here we prove a similar result for CGLMP$_4$: a moderately high violation can certify $d\geq 3$, and a still larger violation would certify $d\geq 4$.

We also present a proof-of-principle experimental implementation using the orbital angular momentum (OAM) degree of freedom of entangled photon pairs. In this experiment, the statistics of $d$-outcome measurements are evaluated by performing $d$ rank-1 measurements. We report a violation sufficient to certify that at least four-dimensional entanglement is present. The dimension of entangled systems has been discussed in previous experiments using DWs that are not device-independent but rely on rather detailed knowledge of the degrees of freedom involved \cite{Dada2011,Howell2012,Zeilinger1,Zeilinger2}. To our knowledge, ours is the first report of an implementation of a device-independent DW based on the violation of a Bell inequality beyond qutrits.

The plan of the paper is as follows. In Sec.~\ref{sec:theory}, we start by reviewing the Collins-Gisin-Linden-Massar-Popescu family of Bell inequalities, on which our DW is based. Focusing on the case of four outcomes, we approach the maximal violation attainable with two entangled qutrits with a numerical optimisation. This optimisation, though reliable, remains based on heuristics. Therefore, in Sec.~\ref{sec:theory2}, we present a bound based on the negativity: it is slightly more demanding than the previous one, but it guarantees a rigorous conclusion. Finally, the experimental set up and procedure are described in Sec.~\ref{sec:experiment}, followed by the results and their discussion.

\section{The inequality and a numerical bound}
\label{sec:theory}
Let us consider two separated parties Alice and Bob, who can share a quantum state and choose to measure it locally in one of two possible ways, obtaining one out of 4 possible outcomes. We denote the measurement setting of each party as $x,y \in \{0,1\}$, and their outcome $a,b \in \{0,1,2,3\}$, so that their measurement statistics can be summarized in the joint probability distribution $\prob{a,b|x,y}$. A Bell inequality satisfied by local probability distributions in this scenario is the CGLMP$_4$ inequality, that we use in the form of Eqn.~(41) of~\cite{CG}:
\begin{eqnarray}
\label{cglmp}
I_4 = &\prob{a \leq b|0,0} + \prob{a \geq b|0,1} + \nonumber\\
&\prob{a \geq b|1,0} - \prob{a \geq b|1,1} - 2 \leq 0.
\end{eqnarray}
The maximum quantum value $I_4^{Q}=\Big[\sqrt{4\sqrt{2-\sqrt{2}}-3\sqrt{2}+8}+\sqrt{2+\sqrt{2}}-3\Big]/4\simeq 0.365$ can be achieved by measuring a quantum state of dimension four (see \ref{sec:ququart}).

In order to show that $I_4$ is a valid dimension witness for quantum dimension 4, we need to find the maximum value $I_4^{(3)}$ that can attain upon measurement of a qutrit state with arbitrary measurements. To date, there is no known way of calculating this bound exactly. One can approximate the bound from below using nonlinear numerical optimisations, with the danger however of finding only a local maximum and thus of reaching wrong conclusions. Alternatively, techniques developed in the context of device-independent quantum information provide upper bounds, which are certainly valid but may not be tight. Here we describe both approaches, starting with the first.

A full implementation of the numerical optimisation would require parametrising all the possible four-outcome POVMs on two qutrits, and there is no known efficient way of doing that. Thus, we first restricted to projective measurements, whose three outcomes are later post-processed classically into four. With this restriction, we find analytically (see \ref{sec:qutrit}):
\begin{eqnarray}\label{eq:boundPOVM}
I_4^{(3)} \geq I_{P} = \frac{\sqrt{33}-3}{9} \approx 0.30495.
\end{eqnarray}
Extensive numerical search using the see-saw algorithm~\cite{liang2009reexamination,pal2010maximal} did not find any improvement, which suggests that this is the bound of violation for qutrits. However, it would be desirable to have a more rigorously proved bound. For this, we turn to the second approach, which provides device-independent bounds.

\section{Bounds from negativity}
\label{sec:theory2}

Several techniques have been proposed to obtain upper bounds on the value of Bell inequalities with fixed dimension~\cite{Moroder,Navascues2014Characterization}. The method we follow is built out of two observations. The first observation is that the measure of entanglement called \textit{negativity}~\cite{zyczkowski1998,vidal2002} is a DW, since for a two-qu$d$it state, one has ${\cal N}\leq \frac{d-1}{2}$. The second observation is that a lower bound on ${\cal N}$ can be obtained in a device-independent way, since the problem can be cast as a semi-definite programme with no assumptions on the structure of the quantum state and measurements~\cite{Moroder}. Concretely, what one does is to find the minimal possible value of ${\cal N}$ conditioned on the value of $I_4$: if ${\cal N}_{min}(I)>\frac{3-1}{2}=1$, the violation $I$ cannot have been obtained with qutrit states. The result is going to be the following: a two-qutrit state certainly cannot violate the CGLMP inequality by more than
\begin{eqnarray}\label{eq:boundSDP}
I_4^{(3)} \leq I_{N} \approx 0.315.
\end{eqnarray}
A violation of this bound thus certifies that entanglement is present in the measured state in a Hilbert space of dimension 4 or higher.

Let us now prove the claim of Eqn.~\eqref{eq:boundSDP}. Negativity is defined as:
\begin{eqnarray}
\mathcal{N} (\rho) = \frac{\| \rho^{\Gamma_A}\| - 1}{2} = \sum_i \frac{|\mu_i| - \mu_i}{2},
\end{eqnarray}
where $\mu_i$ are the eigenvalues of $\rho^{\Gamma_A}$, the partial transposed $\rho$. In general, the negativity of a $d$-by-$d$ quantum state is bounded by $\frac{d-1}{2}$. This is easy to see for pure bipartite $d$-dimensional state written in their Schmidt form $ \ket{\psi} = \sum_{i=0}^{d-1} \sqrt{\lambda_i} \ket{ii}$
with $\sum_i \lambda_i =1$. For these states, the negativity is $\mathcal{N}(\ket{\psi}\bra{\psi}) = \sum_{i\neq j} \sqrt{\lambda_i\lambda_j} $. By the method of Lagrange multiplier, maximizing this expression subject to $\sum_i \lambda_i =1$, attains extremal value when $\lambda_i = \frac{1}{d}$ for all $i$, thus giving the maximal negativity of $\frac{d-1}{2}$. In the case of mixed states, the argument follows by convexity of the negativity. Thus, a lower bound on the negativity of a state thus also puts a lower bound on its dimensionality.

In order to bound the negativity of a quantum state from its observed statistics, we introduce the matrix of moments $\chi$ at local level $\ell$~\cite{navascues2008convergent,Moroder}. For some quantum state $\rho_{AB}$ and measurements $M_{a|x}^A$ and $M_{b|y}^B$, this matrix is defined as:
\begin{eqnarray}
\chi [\rho] = \sum_{i,j,k,l} \ket{ij}_{\bar{A}\bar{B}}\bra{kl} \chi ^{kl}_{ij},
\end{eqnarray}
where $\chi^{kl}_{ij} = \Tr{\rho_{AB} A_{\bar{k}} ^\dagger A_{\bar{i}} \otimes B_{\bar{l}} ^\dagger B_{\bar{j}}}$, and $A_{\bar{i}} = A_{i_1}A_{i_2} \cdots A_{i_\ell}$ is a product of $\ell$ operators chosen from the set of identity and projectors of measurements, $\{ \Id, M_{a|x}^A \}$, similarly for $B_j$'s. Here, $i,j$ indicates and rows, while $k,l$ indicates the column of $\chi$. By construction, $\chi$ can be viewed as a local processing of the original state $\rho$, hence $\mathcal{N}(\chi[\rho]) \leq \mathcal{N}(\rho)$, lower bounding $\mathcal{N}(\chi[\rho])$ also lower bounds $\mathcal{N}(\rho)$. It is thus sufficient to bound the negativity of the moment matrix $\chi$ in order to bound that of $\rho$. As negativity can be formulated as a trace~\cite{vidal2002}, together with the constraint on the observed violation, bounding negativity can be done by solving the following semidefinite programme~\cite{Moroder}:
\begin{align}\label{sdp}
\mathcal{N}(\chi[\rho])\ \geq\ \min_{\chi,\sigma_+,\sigma_-} 
\quad &
\Tr{\sigma_-}\\
\text{s.t.} 
\quad &
\chi_{\bar{A}\bar{B}} = \sigma_+ - \sigma_- \geq 0, \nonumber\\
&\sigma_{\pm}^{\Gamma_A} \geq 0, \nonumber\\
&I[\chi] = I_4. \nonumber
\end{align}
Here $I[\chi]$ is the expected value of the CGLMP inequality for the correlations issued from the $\chi$ matrix and $\sigma_\pm$ are two moment matrices of the same form as $\chi$.

Solving \eqref{sdp} yields a bound on the dimension of the quantum system responsible for the observed violation $I_4$ which gets better with the level of relaxation $\ell$. However this computation becomes quickly intractable for increasing level $\ell$ due to the large number of variables contained in the $\chi$ matrix. In order to reduce the number of independent variables, we thus make use of a well-known depolarization procedure (described in~\ref{appdep}) that keeps CGLMP violation unchanged. The effect of depolarization amounts to relabelling of inputs and outcomes, which can be taken into account in the moment matrix $\chi$ by applying some permutation, $\mathcal{D}$, on the rows and columns of $\chi$. Regard the indices $i,\,j,\,k,\,l$ as function of local inputs and outcomes, i.e. $i = i(a,x)$, then the matrix after relabelling is:
\begin{eqnarray}
\label{eqn:perm}
\mathcal{D}(\chi)^{kl}_{ij} = \chi ^{f(k)g(l)}_{f(i)g(j)},
\end{eqnarray}
where $f$ and $g$ are bijective maps from the index space to itself. Since relabelling are local, the maps act on Alice's and Bob's indices independently. Moreover the same total map applies to the columns and rows. Hence we can use the following lemma:

\begin{Lemma}
Let $\mathcal{D}$ be some permutations of the form \eqref{eqn:perm}, then the following are true:
\begin{enumerate}
\item $\chi \geq 0 \Rightarrow \mathcal{D}(\chi) \geq 0$, 
\item $\chi^{\Gamma_{\bar{A}}} \geq 0 \Rightarrow \mathcal{D}(\chi)^{\Gamma_{\bar{A}}} \geq 0.$
\end{enumerate}
\end{Lemma}

\begin{proof}
The first line follows from the fact that permuting rows and columns in the same manner does not alter the eigenvalues of a matrix. If the original $\chi$ is positive, so is the permuted one $\mathcal{D}(\chi)$. To show the second part, observe that:
\begin{eqnarray}
 \left(\mathcal{D}(\chi)^{\Gamma_{\bar{A}}} \right)_{ij}^{kl} = \mathcal{D}(\chi)_{kj}^{il} = \chi _{f(k)g(j)}^{f(i)g(l)} = \left(\chi^{\Gamma_{\bar{A}}}\right)_{f(i)g(j)}^{f(k)g(l)} = \mathcal{D}(\chi^{\Gamma_{\bar{A}}})_{ij}^{kl},
\end{eqnarray}
i.e. the partial transposition $^{\Gamma_{\bar A}}$ commutes with the permutation $\mathcal{D}$. Thus by the first line and $\chi^{\Gamma_{\bar A}} \geq 0$, the second line is proven.
\end{proof}

Since the value of the CGLMP inequality is invariant under the considered depolarization, $I[\chi]=I[\mathcal{D}(\chi)]$. Moreover, by the lemma all other quantities in \eqref{sdp} are conserved by the depolarization. For any solution of the programme~\eqref{sdp}, there exist another solution which is invariant under depolarization. One can thus solve~\eqref{sdp} with such matrices from the start, reducing the number of free parameters, and making the optimization tractable.

With this simplification at hand, we could compute a lower bound on the negativity necessary in order to achieve a given violation of the CGLMP$_4$ inequality. The result is plotted in FIG.~\ref{negativity}. We observe that in order to reach violation above 0.315, the minimum negativity needs to exceed $1$, which is the maximum of a qutrit state could achieve. Hence a violation of the CGLMP$_4$ inequality beyond 0.315, indicates the presence of entanglement in more than three level systems. The qubit bound $\gtrsim 0.21$ corresponding to negativity $1/2$ can also be read off the plot, although a tighter bound $1/\sqrt{2} - 1/2 \approx 0.2071$ has been shown in~\cite{Moroder}.

\begin{figure}[h!]
\begin{center}
\includegraphics[width=0.5\textwidth,natwidth=350,natheight=280]{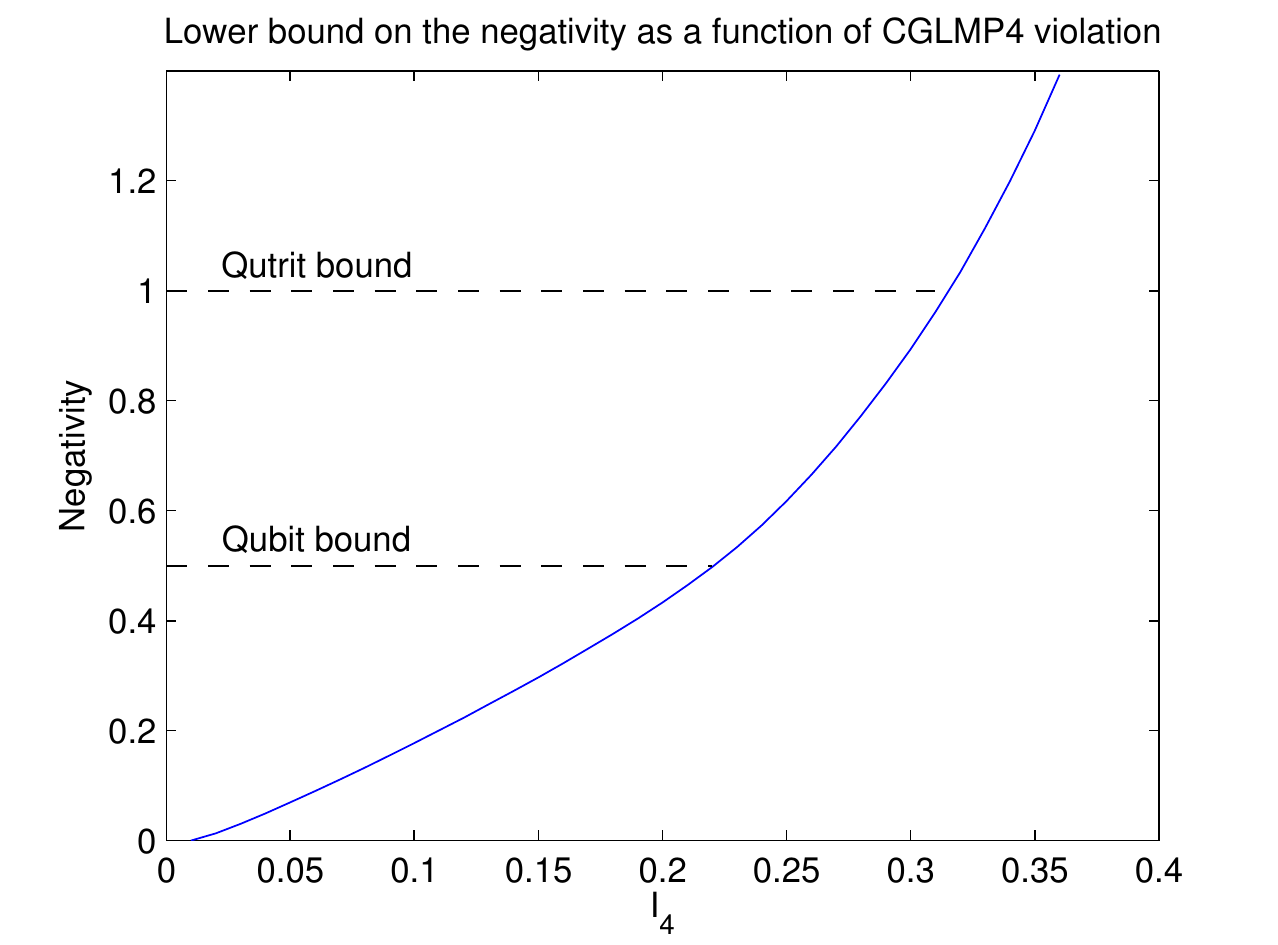}
\caption{Lower bound on the negativity as a function of the CGLMP violation found by solving \eqref{sdp} at local level 2.}
\label{negativity}
\end{center}
\end{figure}

\section{Experiment}
\label{sec:experiment}
Recent experiment by Dada et al. \cite{Dada2011} has demonstrated violation of CGLMP inequalities with photons entangled in orbital angular momentum. The orbital angular momentum (OAM) of a photon in OAM state $|\ell\rangle$, is associated with the helical phase structure, $e^{i\ell\phi}$, where $\ell\hbar$ is the OAM of the photon, and $\phi$ is the azimuthal angle \cite{Allen1992}. Because $\ell$ can take on any integer value, the OAM state space has great potential for high-dimensional entanglement. Note that even though~\cite{Dada2011} measured the $I$ parameter, their reported value is too close to our bound $I_{N}$ to conclude about the dimensionality of the system being maximal in a device-independent manner with satisfying statistical confidence. Here, we evaluate this quantity more precisely, with both the maximally entangled states (MES), and the states achieving the maximum violation (MVS).

\begin{figure}[h!]
\begin{center}
\includegraphics{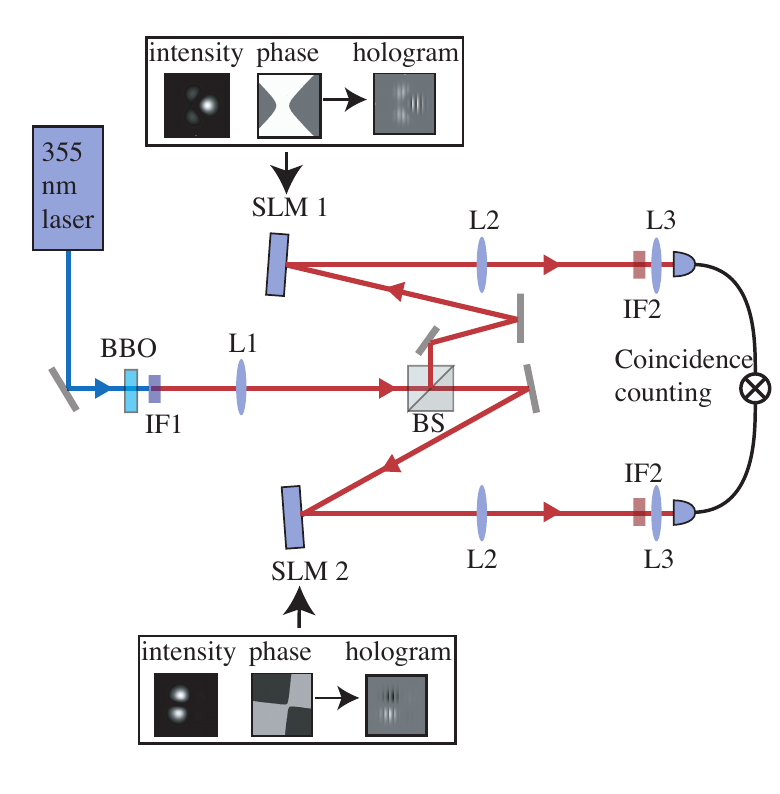}
\caption{Experiment Setup.  The holograms for measuring the states are programmed in SLMs.  The insets show sample intensity and phase cross-sections of the measurement states and the holograms we use to measure. Black to white corresponds to: 0 to 1 for the normalised intensity, 0 to $2\pi$ for the phase and 0 to 255 for the holograms.}
\label{setup}
\end{center}
\end{figure}

High-dimensional OAM entanglement is naturally present in the photons coming from spontaneous parametric down-conversion (SPDC). The generated OAM state is naturally non-maximally entangled  because of the finite crystal size and apertures in the system.  However, phase-matching conditions can be adjusted by tilting the nonlinear crystal used for SPDC, such that the degree of entanglement of OAM states in a four-dimensional subset of the generated OAM-entangled state can be tuned \cite{RomeroSB2012}. This effectively allows us to scan through the parameter $\theta$ which characterises the degree of entanglement in the prepared state
\begin{eqnarray}
\label{eqn:optimalstate}
	\ket{\psi} &=& \frac{1}{\sqrt{2}} \left( \cos \theta \ket{00} + \sin \theta \ket{11} + \sin \theta \ket{22} + \cos \theta \ket{33}
\right).
\end{eqnarray}
%where $\ket{0},\ket{1},\ket{2},\ket{3}$ denotes the canonical basis.
The photon pairs in our experiment come from a 5 mm-long $\beta$-barium borate (BBO) crystal cut for type-1 collinear SPDC.  The crystal is mounted on a fine-control rotation stage to facilitate changing the orientation of the crystal, for changing phase-matching. The crystal is pumped by a collimated 355 nm pump beam (FIG.~\ref{setup}), which is blocked by a longpass filter (IF1) after the crystal.  The signal and idler photons are separated  by a beam splitter and imaged by lenses L1 (f=200 mm) and L2 (f=400 mm) to separate spatial light modulators (SLMs) that act as programmable devices for encoding our measurement states. The SLMs are imaged onto the facets of single-mode fibres by lenses L3 (f=600 mm) and L4 (f=2.0 mm). The single-mode fibres are coupled to avalanche photodiodes (APD) for single photon detection.  To ensure we measure signal and idler photons near degeneracy, we put bandpass filters (IF2) of width 2 nm and centred at 710 nm placed in front of the fibres. The outputs of the APDs are fed to a coincidence circuit (with a timing window of 10 ns)  and the coincidence rate is recorded as a function of the states we specify in the SLM.  For all cases, the states measured (in the computational basis) are given by the optimal measurements:% in (\ref{OptimalMeasurements}).
\begin{eqnarray}
\label{OptimalMeasurements}
%\nonumber 
A_{x} &\equiv& \lbrace \Psi_{x}(a) \rbrace _{a=0}^{d-1}, \quad
\Psi_{x}(a) = \sum_{k=0}^{d-1} \frac{e^{i(2\pi/d)ak}}{\sqrt{d}}(e^{ik\phi_{x}} \ket{k}), \label{OMA}
\\
%\nonumber
 B_{y} &\equiv& \lbrace \Phi_{y}(b) \rbrace _{b=0}^{d-1}, \quad
\Phi_{y}(b) = \sum_{k=0}^{d-1} \frac{e^{-i(2\pi/d)bk}}{\sqrt{d}}(e^{ik\theta_{y}} \ket{k}). \label{OMB}
\end{eqnarray}
It is necessary to modulate both phase and intensity to achieve this, and we follow the hologram design in \cite{Jack2010Pole, Davis1999Encoding}. 

We first orient the crystal such that we get a maximally entangled state (FIG.~\ref{resultsOAM}, B-MES). For this MES case, we work on the subspace spanned by OAM states $\{|-2\rangle, |-1\rangle, |1\rangle, |2\rangle\}$.  The computational basis $\{|0\rangle, |1\rangle, |2\rangle, |3\rangle\}$ corresponds to OAM states $\{|2\rangle, |1\rangle, |-1\rangle, |-2\rangle\}$ in the signal photon, and  $\{|-2\rangle, |-1\rangle, |1\rangle, |2\rangle\}$ in the idler photon.   We obtain a value of $I_4=0.333\pm0.007$ (point B in FIG. \ref{resultsOAM}), higher than the bound for qutrits. 
\begin{figure}[h!]
\begin{center}
\includegraphics{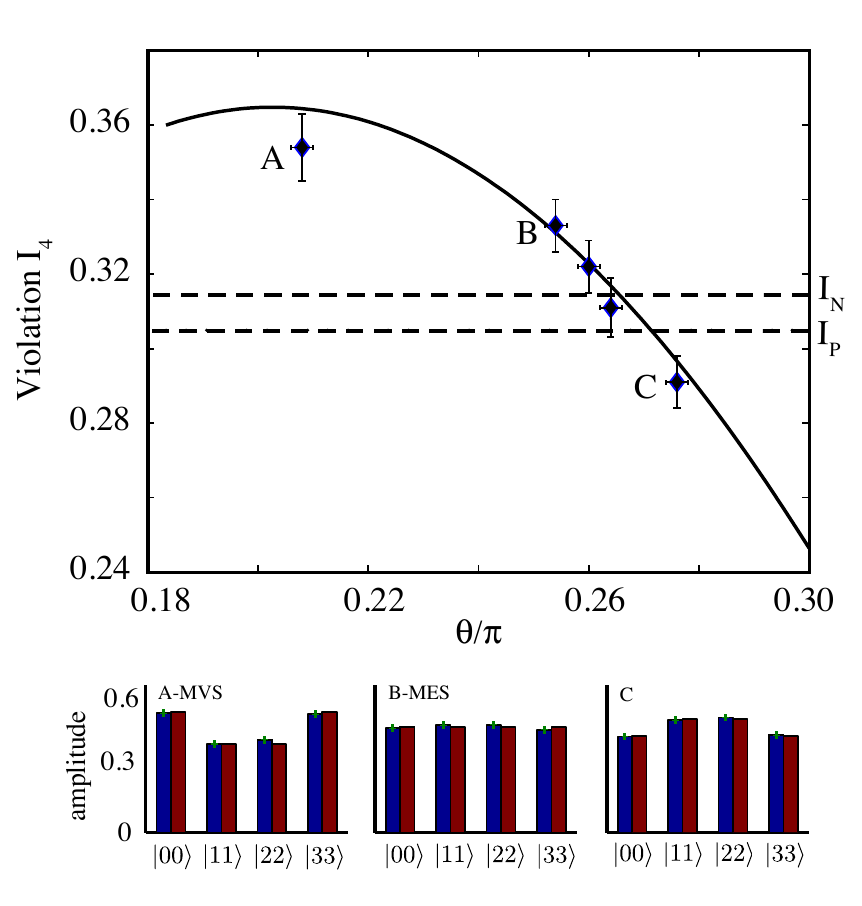}
\caption{Violation of the CGLMP inequality $I_4\leq 0$ using OAM. The solid black line shows the theoretical violations, as a function of the parameter $\theta$ as given in Eqn.\eqref{eqn:optimalstate}.  We highlight three points corresponding to maximal violation state (MVS) (A), maximally entangled state (MES) (B) and no violation of the bounds for qutrits (C) cases. The corresponding states for these points are also shown. The dashed line represent the lower bound $I_P$ and the upper bound $I_N$ on $I_4^{(3)}$.}
\label{resultsOAM}
\end{center}
\end{figure}

For the state that will violate the inequality maximally (FIG.~\ref{resultsOAM}, A-MVS), we work on the subspace spanned by $\{|-4\rangle, |-1\rangle, |1\rangle, |4\rangle\}$. The states of the computational basis, $\{|0\rangle, |1\rangle, |2\rangle, |3\rangle\}$ corresponds to OAM states $\{|4\rangle, |1\rangle, |-1\rangle, |4\rangle\}$  in the signal photon and $\{|-4\rangle, |-1\rangle, |1\rangle, |4\rangle\}$ in the idler photon.  We obtain a value of $I_4=0.354\pm0.009$ (point A in FIG. \ref{resultsOAM}), again higher than the bound for qutrits. 

The solid black line shows the violations obtained from theory.  We also show other experimental violations, which follow the theoretical curve closely. Point C in FIG.~\ref{resultsOAM} is an example where there is no violation of the bound for qutrits, $I_4=0.291\pm0.007$. We obtained this by working on a subspace spanned by OAM states  $\{|-5\rangle, |-1\rangle, |1\rangle, |5\rangle\}$. 

For device independent application of the CGLMP inequality, one would require genuine $d$-outcome measurements. In this experiment, however, only coincidences of rank-1 projectors corresponding to each measurement outcome are measured individually. The joint probability of $d$ outcome measurements, which would have been revealed had the CGLMP measurements been performed, is then reconstructed. One potential method to overcome this limitation is by the application of mode sorters~\cite{mode-sorter}, where one can sort OAM modes into different lateral positions.

\section{Conclusion}
The quantum dimension of a system can be certified in a device-independent way by implementing a dimension witness built on a Bell inequality. This demonstrates both that the system has a large dimension and its quantum nature. We showed that the CGLMP inequality with four outcomes provides such a witness that certifies $d\geq 4$: any violation larger than 0.315 cannot be attributed to smaller-dimensional systems, and is unreachable for classical systems of arbitrary dimension. We have reported such large violations with photons entangled in the orbital angular momentum degree of freedom.

By mastering higher-dimensional quantum systems, one can enhance the performance of quantum communication protocols. The dimension captures how much information can be possibly coded in each carrier, but this is only a first step. For instance, in standard (i.e. not device-independent) QKD, it is known that higher-dimensional protocols are also more robust against noise~\cite{bechmann-pasquinucci,cerf2002security, sheridan2010security,durt2004security}; and several groups have reported experiments in this direction~\cite{Walborn2006QKD,DOQKD,twisted}. For device-independent QKD with higher alphabets, much less is known. A basic study of security against no-signalling adversaries was given in \cite{scarani2006secrecy}. In~\cite{mpa11}, the authors introduce general tools to deal with security of device-independent QKD; then, among the examples, they compute Eve's information for a protocol based on the 3-outcome CGLMP inequality. The theoretical tools and the experimental capability presented in this paper will hopefully trigger significant developments in this direction.

\section*{Acknowledgements}
We acknowledge Yeong-Cherng Liang for useful discussions and Tam\'as V\'ertesi for his help in writing the see-saw algorithm used in this work. J.R. would like to thank Miles Padgett for his support for the experiment.

This work is funded by the Singapore Ministry of Education (partly through the Academic Research Fund Tier 3 MOE2012-T3-1-009) and the National Research Foundation of Singapore, Prime Minister’s Office, under the Research Centres of Excellence programme.

\newpage
\appendix

\section{CGLMP inequality and depolarization}\label{appdep}
In this section we discuss the classical processing, called depolarization, that reduces the number of variables in the probability distribution while maintaining the violation of the Bell inequality. 

In an experiment with two inputs and $n$ outcomes, the joint probability distribution of the outcomes may be recorded as $\prob{a,b|x,y}$, where $a,b \in \left\{0,1,2,\ldots,n-1\right\}$ denote the outcomes, and $x,y \in \left\{ 0,1 \right\}$ denote the choice of measurements. We require $\prob{a,b|x,y}$ to be a proper distribution, i.e. non-negative and normalized, and to respect the no-signalling condition due to the separation between the parties.

%[Fix the analysis to 4 outcomes]

%\begin{subeqnarray}
%\sum_b \prob{a,b|x,0} = \sum_b \prob{a,b|x,1} = \prob{a|x},  \\
%\sum_a \prob{a,b|0,y} = \sum_a \prob{a,b|1,y} = \prob{b|y} .
%\end{subeqnarray}

The joint probabilities can be organised in an $2n$-by-$2n$ array:
\begin{align}
P=
{
\left(
\begin{tabular}{c|c}
\prob{a,b|0,0}& \prob{a,b|0,1} \\
\hline
\prob{a,b|1,0}&\prob{a,b|1,1} \\
\end{tabular}\right).
}
\end{align}
%where $\prob{a,b|x,y}$ are $n$-by-$n$ arrays denoting the joint probability distributions. Note that not all the entries of the table can be chosen freely, as they must satisfy the normalization and no-signalling condition.

%\subsection{The inequality and depolarization}
%\label{sec:depolarization}
The inequality we are interested in here is the so-called CGLMP inequalities~\cite{CGLMP}. In a form as Eqn.(41) of Ref.~\cite{CG}, it can be expressed as

\begin{align}
I_n = \langle \mathcal{I}_n, P \rangle -2 \leq 0,
\label{eq:CGLMPn}
\end{align}
with
{
\renewcommand{\arraystretch}{1.2}
\begin{align}
\label{eqn:CGLMPn}
\mathcal{I}_n = 
\left(
\begin{tabular}{c|c}
\textbf{J}$_n$  & \textbf{J}$^T_n$ \\
\hline
\textbf{J}$^T_n$ & -\textbf{J}$^T_n$ 
\end{tabular}
\right),
\end{align}
}
where \textbf{J} is a n-by-n array with only an upper triangular array filled with 1, $^T$ is the transposition, and $\langle \cdot , \cdot \rangle$ denotes the sum of term-by-term multiplication. With local hidden variables, $I_n \leq 0$, while the generalized PR-box violates up to $I_4^\text{PR} = \frac{n-1}{n}$~\cite{barrett2004nonlocal}.

%% depolarization
%Due to the symmetry presented in the inequality, with a classical post processing of the statistics, we can reduce the number of free parameters in $P$ down to three without changing the violation. This post processing, called depolarisation, which consists of essentially relabelling inputs and outcomes is described in detail in~\cite{scarani2006secrecy}.
%
Due to the symmetry present in the inequality, the number of parameters in $P$ that are relevant for the value of $I_n$ can be reduced. Indeed, there exists a classical post processing, also called a depolarization, that maps all the points in the probability distribution space to a lower dimension space, while keeping the CGLMP violation unchanged \cite{scarani2006secrecy}. Under the action of this map, every probability distribution is projected to a slice of the no-signalling polytope. This procedure can be described as follows:

\textbf{Step 1.} Alice and Bob add a number $k$, uniformly chosen from $\{0,\cdots , n-1\}$, to their outcomes:
\begin{eqnarray}
\nonumber a &\rightarrow a + k, \;b &\rightarrow b + k.
\end{eqnarray}

\textbf{Step 2.} With probability $\frac{1}{4}$ according shared randomness, Alice and Bob perform one of the four possible processes:
\begin{eqnarray}
\nonumber	\textbf{Proc 1.} &A: \text{Do nothing,} 
& B: \text{Do nothing;} \\
\nonumber	\textbf{Proc 2.} &A: x \rightarrow \bar{x},a\rightarrow-a, 
& B: b \rightarrow -b+y; \\
\nonumber	\textbf{Proc 3.} &A: a \rightarrow -a-x,
& B: y\rightarrow\bar{y},b\rightarrow-b; \\
\nonumber	\textbf{Proc 4.} &A: x\rightarrow\bar{x},a\rightarrow a+x,
& B:  y\rightarrow\bar{y},b\rightarrow b+ \bar{y};  
\end{eqnarray}
where $\bar{x} = 1-x$ and the operation on the outcome are done modulo $n$. These implements $\text{P} \rightarrow \text{P'}$ such that it only depends on the difference of the outcome $\Delta = a-b$ as follows:
%\begin{eqnarray}
%\label{eqn:depolarization1}
%\nonumber 4P_2(a,b|0,0) =&P_1(a,b|0,0) + P_1(-a,-b|0,1) \\
%\nonumber &+P_1(-a,-b|1,0) + P_1(a,b+1|1,1), \\
%\label{eqn:depolarization2}
%4\nonumber P_2(a,b|0,1) =&P_1(a,b|0,1) + P_1(-a,-b+1|1,1) \\
%\nonumber&+P_1(-a,-b|0,0) + P_1(a,b|1,0), \\
%\label{eqn:depolarization3}
%4\nonumber P_2(a,b|1,0) =&P_1(a,b|1,0) + P_1(-a-1,-b|1,1) \\
%\nonumber&+P_1(-a,-b|0,0) + P_1(a+1,b+1|0,1), \\
%\label{eqn:depolarization4}
%4\nonumber P_2(a,b|1,1) =&P_1(a,b|1,1) + P_1(-a,-b+1|0,1) \\
%\nonumber&+P_1(-a-1,-b|1,0) + P_1(a+1,b|0,0),
%4P_2(a,b|x,y) = P_1(a,b|x,y)+P_1(-a,-b+y|\bar{x},y)+P_1(-a-x,-b|x,\bar{y})+P_1(a+x,b+\bar{y}|\bar{x},\bar{y}).
%\end{eqnarray}
%
%From $P_2$, we can see the relation of $P_2(\Delta|x,y)$ with different $x$ and $y$ are:
\begin{align}
\label{eqn:slice}
	P'(\Delta|0,0)=P'(-\Delta|0,1)=P'(-\Delta|1,0)=P'(\Delta+1|1,1).
\end{align}
The number of free variables is thus reduced to $n-1$.

%\subsection{Four-outcomes case}
In the case of $n=4$ outcomes, the depolarised probability takes the following form:
\begin{eqnarray}
P\xrightarrow{\text{depolarisation}} P' = \frac{1}{4}
\left(\begin{array}{cccc|cccc}
		
										p_0	&p_3  &p_2	&p_1  &p_0	&p_1	&p_2	&p_3\\
										p_1	&p_0  &p_3  &p_2	&p_3	&p_0	&p_1	&p_2\\
										p_2	&p_1  &p_0  &p_3	&p_2	&p_3	&p_0	&p_1\\
										p_3	&p_2	&p_1  &p_0	&p_1	&p_2	&p_3	&p_0\\
		\hline
										p_0	&p_1	&p_2	&p_3	&p_1	&p_0	&p_3	&p_2\\
										p_3	&p_0	&p_1	&p_2	&p_2	&p_1	&p_0	&p_3\\
										p_2	&p_3	&p_0	&p_1	&p_3	&p_2	&p_1	&p_0\\
										p_1	&p_2	&p_3	&p_0	&p_0	&p_3	&p_2	&p_1\\
\end{array}\right),
\end{eqnarray}
with $\sum p_i = 1$. %This reduction of number of variables is crucial in the numerical optimization processes. 
%The number of variables now scales $O(n)$ instead of $O(n^2)$ that makes computations more tractable.
Since these probabilities live in a three-dimensional space, they can be easily represented geometrically (see FIG.~\ref{fig:depolarizedSpace}).

\begin{figure}[tb]
\begin{center}
\includegraphics[width=0.6\textwidth]{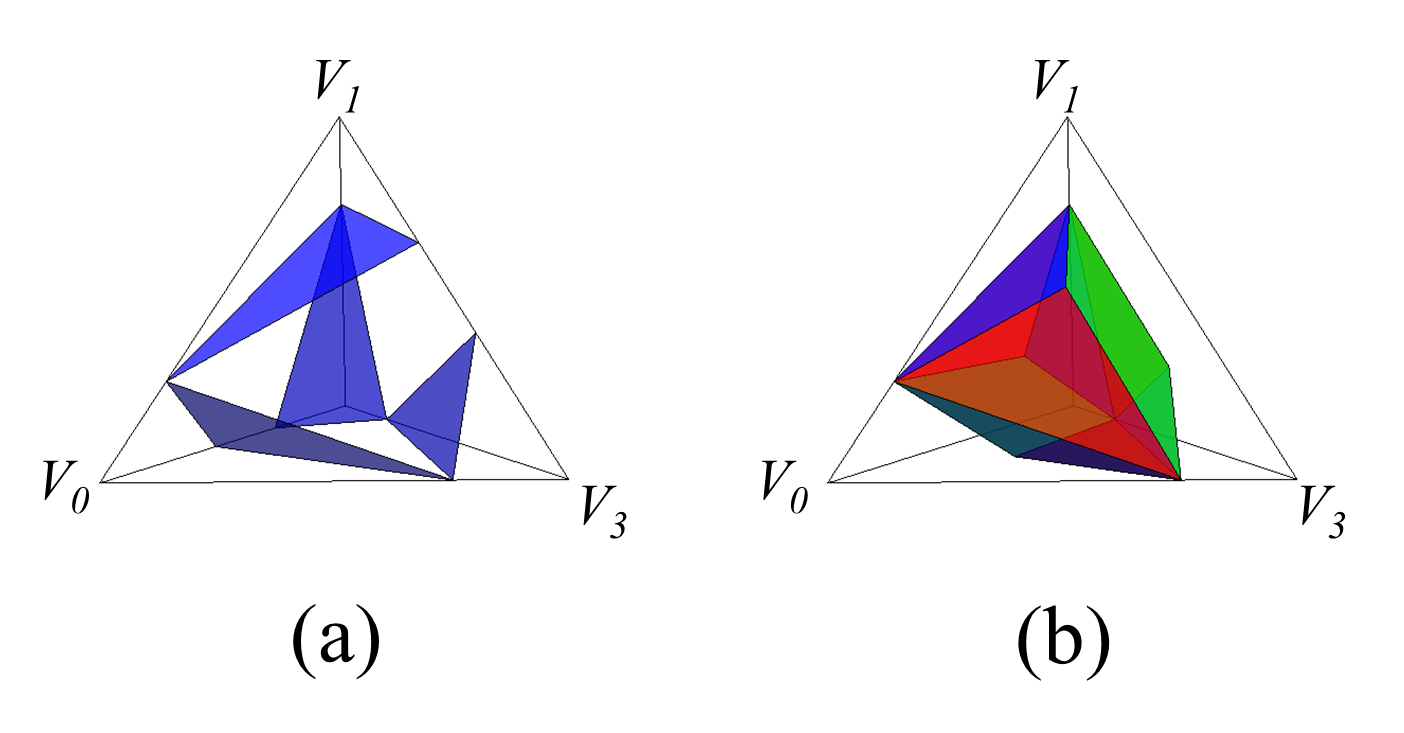}
\caption{(Colour online) Geometrical representation of the 2-settings 4-outcomes depolarized probability space. The no-signalling polytope is the pyramid with extremal points $V_i,i\in\{0,1,2,3\}$, where $V_i$ is a generalized PR box $p_i=1$. (a) Different relabelling of the CGLMP inequality impose four constraints on local correlations in this space. (b) The local polytope is fully determined here by three kinds of facets (blue, green, red). The blue facets represents the CGLMP inequality considered in the main text.}
\label{fig:depolarizedSpace}
\end{center}
\end{figure}

\section{Violation of CGLMP$_4$ inequality with four level systems}
\label{sec:ququart}
In this section, we discuss the violation of CGLMP$_4$ inequality with four-level systems, with focus on maximizing this violation. Incidentally, the maximal violation with four-level systems is also the maximal quantum violation of this inequality.

Zohren and Gill~\cite{Zohren2008Maximal} showed that for $n \geq 3$ and when the dimension of the system is same as the number of outcomes, the state that allows for a maximal violation of the CGLMP inequality is not the maximally entangled state, but a partially entangled one. This was already conjectured for small $n$ in~\cite{acin2002quantum}, and recently confirmed in the case $n=3$ through self-testing~\cite{yang2014robust}. The optimal measurement basis are conjectured to be the same for both the maximally entangled state and the maximal violation state, namely:

\begin{eqnarray}
%\label{OptimalMeasurements}
% \nonumber 
A_{x} &\equiv& \lbrace \Psi_{x}(a) \rbrace _{a=0}^{d-1}, \quad
\Psi_{x}(a) = \sum_{k=0}^{d-1} \frac{e^{i(2\pi/d)ak}}{\sqrt{d}}(e^{ik\phi_{x}} \ket{k}), \label{OMAappen}
\\
% \nonumber
 B_{y} &\equiv& \lbrace \Phi_{y}(b) \rbrace _{b=0}^{d-1}, \quad
\Phi_{y}(b) = \sum_{k=0}^{d-1} \frac{e^{-i(2\pi/d)bk}}{\sqrt{d}}(e^{ik\theta_{y}} \ket{k}), \label{OMBappen}
\end{eqnarray}
where $\phi_0=0, \phi_1=\frac{\pi}{d}, \theta_0 = -\frac{\pi}{2d}$ and $\theta_1 = \frac{\pi}{2d}$.

Let us restrict ourselves to the case $n=4$, to find the maximal violation of the $I_4$ inequality. For this we note that after fixing the measurement settings, finding the maximal violation amounts to find the maximum eigenvalue and corresponding eigenvector of the Bell operator. Another approach is based on the observed symmetries in the states that optimizes the violation. In the Schmidt form, numerical optimization indicates that the maximal violation states are of the following form:
\begin{eqnarray}
\label{eqn:optimalstateAppen}
	\ket{\psi} &=& \frac{1}{\sqrt{2}} \left( \cos \theta \ket{00} + \sin \theta \ket{11} + \sin \theta \ket{22} + \cos \theta \ket{33} \right).
\end{eqnarray}
For state \eqref{eqn:optimalstateAppen} and measurements (\ref{OMAappen}-\ref{OMBappen}) the CGLMP value $I$ becomes a function of the single parameter $\theta$:
\begin{eqnarray}
\label{eqn:curve}
I(\theta) = 
%-\frac{3}{4} + \left( \frac{\sqrt{2}}{2} + \sqrt{2}S \right) \cos\theta \sin\theta + \frac{\sqrt{2}S}{2}\sin^2\theta + \frac{\sqrt{2}C}{2}\cos^2\theta,
\left( -\frac{3}{4}+ \frac{C}{2}\right) + \left( \frac{1}{2\sqrt{2}} + \frac{S}{\sqrt{2}}\right) \sin{2\theta} + \frac{S}{2} \cos{2\theta},
\end{eqnarray}
where we introduce constants $C=\cos{\frac{\pi}{8}}$ and $S=\sin{\frac{\pi}{8}}$, $C-S = \sqrt{2}S$ and $C+S = \sqrt{2}C$ are used to simplify the expression. Eqn.~\eqref{eqn:curve} is plotted in FIG.~\ref{resultsOAM} to compare with experimental data points.

By setting $\frac{dI}{d\theta} = 0$, the maximal violation state is achieved with $\theta$ such that 
\begin{eqnarray}
%\cos 2\theta = \frac{\sqrt{2}S}{\sqrt{6S^2+4S+1}},
\tan{2\theta} = \frac{1+2S}{\sqrt{2}S}, \; \cos{2\theta} = \frac{\sqrt{2}S}{\sqrt{6S^2+4S+1}}, \; \sin{2\theta} = \frac{1+2S}{\sqrt{6S^2+4S+1}}
\end{eqnarray}
and the maximal violation is:
\begin{eqnarray}
%I^{max} &=& -\frac{3}{4} + \frac{C+yS+\sqrt{\frac{1-y^2}{2}}(1+2S)}{2} \\
%
I^{max} =-\frac{3}{4}+\frac{C}{2}+\frac{1}{2\sqrt{2}}\sqrt{6S^2+4S+1}\approx & 0.364762.
\label{eq:maxCGLMP}
\end{eqnarray}

In order to check whether this is the maximal violation for any quantum states, we compare the violation with upper bounds obtained with the semidefinite programme (SDP) hierarchy of quantum correlations~\cite{navascues2008convergent}: the maximal violation $I^{max}$ agrees with the Table 1 of \cite{navascues2008convergent} up to $10^{-6}$.

%%%%%%%%%%%%%%%%%%%%%%%%%%%%%%%%%%%%%%%%%%%%%%%%%%%%%%%%%%%%%%%%%%%%%%%%%%%%
\section{Violation of CGLMP$_4$ inequality with three level systems}
\label{sec:qutrit}
If we restrict ourself to only entangled qutrits (three level quantum systems), the violation of the inequality \eqref{eq:CGLMPn} might possibly be lower. An upper bound on $I^{(3)}_4$ is derived in the main text via negativity with semidefinite programme. Here a lower bound on $I_4^{(3)}$, $I_p$, is derived by finding a two-qutrit state and measurements which achieve some violation of~\eqref{eq:CGLMPn}. The bound we obtain here is based on a restricted class of POVM, and supported by numeric evidence.

Namely, we consider qutrit measurements that never produce the last outcome. These are measurement with only three possible outcomes. Hence we assign a probability 0 to the last outcome of all measurements:
\begin{eqnarray}
\prob{a,3|x,y} = \prob{3,b|x,y} &=&0
\end{eqnarray}

Due to the null probability of the fourth outcome, the last number in each row and column of each square of $\mathcal{I}_4$ is irrelevant. In particular, the value of $I_4$ is the same as that obtained by exchanging the table of coefficients $\mathcal{I}_4$ by 
{
\renewcommand{\arraystretch}{1.2}
\begin{align}
%\label{eqn:CGLMPn}
\mathcal{I}_4' = 
\left(
\begin{tabular}{cc|cc}
\textbf{J}$_3$  & 0 & \textbf{J}$^T_3$ & 0\\
0& 0& 0& 0\\
\hline
\textbf{J}$^T_3$ & 0 & -\textbf{J}$^T_3$ & 0\\
0& 0& 0& 0
\end{tabular}
\right).
\end{align}
}

Since this table is identical that of the CGLMP inequality with three outcomes, maximizing the $I_4$ inequality with this chosen simple POVMs is thus equivalent to testing the CGLMP$_3$ inequality, $I_3$. The maximum violation of $I_3$ thus constitutes an upper bound on the maximum violation of CGLMP$_4$ achievable with qutrits and this choice of simple POVMs. Namely, this bound is known to be $I_3^*$ to be $\frac{\sqrt{33}-3}{9}\approx 0.30495$.

Note that in the case where the projective measurement is chosen as to forbid a different outcome than the last one, maximizing $I_4$ can amount to maximizing one of 8 possible 3-outcome Bell expressions. The NPA hierarchy allows one to bound the maximal quantum violation of each of these inequalities to either $I_3^*$, $I_2^*$ or 0, of which $I_3^*$ is the largest.

In order to check whether the bound $I_4^{(3)}\leq I_P=I_3^*$ remains valid for general POVM's we turned to numerical optimization method. An iterative numerical optimization procedure, called the see-saw method was introduced in \cite{werner2001} and further developed in \cite{liang2009reexamination}. It can maximize the violation of a Bell inequality with a constraint on the dimension of the measured system, but with no constraint on the measurement used. Although it is not guaranteed to converge, it has given remarkable results in similar contexts \cite{pal2010maximal}. We did not find any violation larger than $I_3^*$ with this method. We thus conjecture that $I_3^*$ is indeed the maximum violation of the inequality $I_4$ achievable with three level quantum systems, even with general POVMs.

\end{document}